\documentclass{article}
\usepackage[dvips]{graphicx}
\usepackage{amsmath}
\usepackage{amssymb}
\usepackage{amsfonts}
\usepackage{amsthm}

\newtheorem{lem}{Lemma}

\newtheorem{prop}{Proposition}

\newtheorem{remark}{Remark}

\title{Horizontal Product Differentiation in Varian's Model of Sales\footnote{I am grateful to Tadashi Sekiguchi for his guidance over the past decade. Without his helpful advice, I would not have been able to complete this study. I am also grateful to Haruo Imai, my supervisor, for his long direction. I also would like to thank Noriaki Matsushima, Takanori Maehara, and seminar participants at Kyoto University, National Graduate Institute for Policy Studies, the ARSC2011(Applied Regional Science Conference), and the NARSC2012(the North American Meeting of Regional Science Association International) for comments and suggestions. All errors are my own.}}
\author{Kuninori NAKAGAWA \footnote{Shizuoka University, Suruga-ku Ohya 836, Shizuoka, 422-8529, Japan. Tel: +81-54-237-5474. E-mail: nakagawa.kuninori@shizuoka.ac.jp ORCID iD 0000-0003-2429-7630}}
\date{\today}

\begin{document}
\maketitle
\begin{abstract}
We consider the explicit introduction of firms' choice of location to Varian's model of sales for a two-stage spatial competition model based on a standard Hotelling's linear city model. This model is the formalization of Varian's model of sales in the context of Hotelling's spatial competition. We obtain three main results. First, we show that there exists a symmetric subgame perfect equilibrium in which each firm chooses a symmetric mixed strategy equilibrium profile. This equilibrium includes symmetric location pairs and asymmetric location pairs. Second, the equilibrium behaviors in our model are randomized at both location and price stages. Third, we show that expected profits in a subgame perfect equilibrium are equal to the maximum monopoly profit. Thus, even when product differentiation is explicitly introduced into a Varian-type model, Varian's implication can be retained; the opportunity for profit in an informed market is lost with competition.
\end{abstract}

\newpage
\section{Introduction}
According to classical theory, temporary price discounts, popularly known as “sales,” should not occur in a market at equilibrium. However, Varian (1980) proposed a model in which such sales can occur in equilibrium in the form of a mixed strategy. This equilibrium price dispersion is widely known for the phenomenon of breaking the law of one price. A key assumption of Varian's model is that some consumers preferentially buy from particular firms without considering other firms; whether such a consumer purchases a unit of goods is determined solely by comparing that consumer's reserve price to the preferred firm's offered price. As a consequence, goods from different firms may be physically identical but may not be identical from the perspective of some consumers. In Varian's model, the heterogeneity among products is generated by asymmetric information: some consumers possess information on the presence of a product that others do not possess. In this article, we present a generalization of Varian (1980) with product differentiation in an explicit way, which is the same as in a spatial competition model of D'Aspremont et al (1979). In other words, we explicitly introduce a firm's choice of location to Varian's model and explore a two-stage game with spatial competition based on a standard Hotelling linear model, (Hotelling 1929).

In our study, we define Varian's model of sales as a model that belongs to a group of Hotelling's location price competition model. Considering the example of Apple fans or members of airlines' FFP (Frequently Flyer Programs), some enthusiastic consumers seek to support a firm's market power. With this type of purchasing behavior, price discrimination between informed and uninformed consumers is a typical economics issue. In oligopoly pricing models, Salop and Stiglitz (1977) and Varian (1980) developed the earliest successful models, which analyzed equilibrium price dispersion in this market structure. Varian focused on inter-temporal changes in price at one firm using mixed strategy. In this article, we analyze a horizontal product differentiation in Varian's model of sales. This motivation is reflected in two structures of our model.

First, as in a Hotelling linear model, consumers in our model are also located on a line segment. However, in our model, consumers are located on three points in the [0,1] interval: both opposite endpoints and the mid-center. We consider a duopoly model. Each firm is monopolist to the consumer group at each end, 0 and 1. In contrast, both firms compete for the consumers at the mid-center. In other words, as in Varian's model, consumers in our model are exogenously assigned to groups: group 1 consumers will only buy from firm 1; group 2 consumers will only buy from firm 2; group 3 consumers will buy from either firm 1 or firm 2 depending on which offers the highest net utility. This structure corresponds to the setting of consumer type in Varian's model of sales. For ease of reference, we call these groups 1-uninformed, 2-uninformed, and informed switchers; 1-uninformed and 2-uninformed consumers will collectively be called uninformed consumers or its fans.

Second, as in a spatial competition model, in our model, firms control two product characteristics: location and price, and each consumer incurs a transportation cost quadratic on the distance from the firm. Each consumer has a reservation price known to the firms and receives utility from the purchase equal to the reserve less the price and transport cost. Furthermore, in our model, each firm chooses its product characteristic near and within a territory where it fits relatively to its uninformed. Firm 1 chooses its location within [0,1/2] and firm 2 chooses its location within [1/2,1], because this model is symmetric with respect to the center at which informed switchers locate. From Varian's point of view, our settings will be interpreted as follows. When the uninformed consumers at either side move to the other side beyond the mid-center, they should be charged for their move with charges that exceed the limit amount of their reservation utility (willingness to pay). In other words, Varian's result depends on the assumption that at least some consumers are not loyal to a particular firm, and they are indifferent to price information offered by firms. In brief, Varian's model implies symmetric patterns of horizontal product differentiation derived from the consumers' preferences described above. In contrast to Varian's model, our model analyzes product differentiation between firms' products. Therefore, goods are heterogeneous even if all consumers know all of the information about the products such as price and other factors. 

In this article, we obtain three main results. First, we show that there exists a subgame perfect equilibrium in which the firms' strategies are symmetric across the mid-center of the line. In the equilibrium, each firm randomly chooses its location. Consequently, various equilibrium location patterns are realized. For example, this mixed-strategy equilibrium includes two typical positions, each firm at either end of the interval (maximum differentiation) and one with both firms at the center (minimum differentiation). In addition, this equilibrium includes symmetric location pairs and asymmetric location pairs.

Second, the equilibrium behaviors in our model are randomized at both location and price stages. This equilibrium in price subgame is typically called price dispersion, which is well known for the equilibrium in Varian's model of sales. However, we consider the equilibrium given explicit location choice. This result shows that the equilibrium price dispersion shown by Varian is not only limited to symmetric location pairs but also extended to asymmetric location pairs.

Moreover, we obtain a result that location choice is also randomized, in other words, equilibrium location is also dispersed in a subgame perfect equilibrium. Our result implies that differentiation between products would also be randomized; that is, the same as the case with price shown by Varian. The cause of ``dispersion'' of product differentiation is quite different from the cause of price dispersion. In the last section, we discuss this in detail.

Third, we show that expected profits in a subgame perfect equilibrium are equal to the maximum monopoly profit. The equilibrium profit of each firm is equal to a reservation value of each firm's uninformed consumer. This result is the same as Varian's equilibrium profit. Our model can be regarded as explicitly introducing product differentiation choice into Varian's sales model so that the result of firms' profits in a subgame perfect equilibrium in our model is the same as the result Varian showed for price competition only. In other words, this result shows that Bertrand's law on the loss of opportunity to gain excess profit in a competitive market still appears even when locational differentiation is explicitly considered. In the last section, we also discuss these results, particularly focusing on the dispersion at both location and price choices.

Our analysis proceeds as follows. We construct a two-stage game: firms choose a location in the first stage and a price in the second stage. We solve this game by backward induction. We analyze price competition on a location-pair subgame. We obtain a mixed strategy equilibrium in the pricing subgame for any pair of locations. Then, we analyze a subgame perfect location equilibrium in the first stage. Finally, in the concluding remarks, we discuss in detail the relationships between Varian's model and pay-off discontinuity in Hotelling's game.

\subsection{Literature}
In our model, we modify a two-stage spatial competition model of D'Aspremont et al. (1979) incorporating the elements of Varian's model of sales (Varian 1980). We consider a model that incorporates a spatial competition framework \`a la Hotelling, where firms compete in price and location, into Varian's model of sales. We address both sides, that is, from Varian's model and Hotelling's model. This section reviews the literature on Hotelling.

In the context of a spatial competition model, based on Hotelling's ``Main Street Model,'' firms' choice of location in an equilibrium involves the opposite of two extreme results, minimum and maximum differentiation.

Typically, a firm's choice of location depends on the severity of the price competition. It is possible that firms are agglomerated if the pressure of price competition is relaxed effectively by some means. We introduce three approaches other than the approach that we developed, which provide the result of agglomeration in the context of spatial competition.

The first approach is firms' collusive behavior with bargaining. For example, if firms establish a cartel and then maintain a higher price than the competitive price, the degree of locational differentiation is minimized. Jehiel (1992) examined this implication in an infinite horizon setting where, in the initial period, the firms choose their location and, in subsequent periods, charge the Nash bargaining solution prices. Then, in the unique equilibrium, both firms are located at the center of the market. Friedman and Thisse (1993) also presented a similar model where firms infinitely repeat bargaining in a price subgame and showed that the degree of differentiation is minimized. Matsumura and Matsushima (2011) showed that introducing a cost difference between two firms causes agglomeration to disappear in a collusive equilibrium. Roth and Zhao (2003) have also shown another result derived from either egalitarian or Kalai-Smorodinski bargaining solutions. The authors showed that there is a continuum of symmetric equilibria where the firms locate apart from each other.

In the second approach, firms compete in a multi-dimensional space. Irmen and Thisse (1998) showed that in the location game with n characteristics, firms choose to maximize differentiation for the dominant characteristic and to minimize differentiation for the other characteristics. Prior to Irmen and Thisse, in the context of a spatial competition model with quadratic transportation costs, Tabuchi (1994) extended this theory to a two-dimensional model. Then, Ansari et al. (1998) analyzed a three-dimensional model. Economides (1986) also studied a two-dimensional characteristic space model. Aoyagi and Okabe (1993) analyzed a two-dimensional model using a numerical analysis method.

In the third approach, Anderson et al. (1992) applied a discrete choice model to the spatial competition model. De Palma et al. (1985) showed that this model derives various types of location pairs in an equilibrium. This model implies that each firm can capture consumers from every point on the line. Because of this, severe competition will not arise even as firms move closer to each other, thereby giving firms the incentive to agglomerate at the center of the line. This can be considered an attempt to extend a consumer's behavior on the Hotelling linear model.

Additionally, in our model, a pay-off function is continuous in the location game but, on the other hand, it is discontinuous in the price subgame. The pay-off discontinuity at a price subgame is widely discussed in this literature. Dasgupta and Maskin (1986) showed that a mixed strategy equilibrium at price subgame exists for this type of location-price competition. Given this result, Osborne and Pitchik (1986, 1987) calculated a subgame perfect location point for this type of game. Bester et al. (1996), Matsumura and Matsushima (2009), and Eaton and Tweedle (2012) also analyzed a mixed strategy equilibrium on Hotelling's location game. Caplin and Nalebuff (1991) showed that quadratic transportation cost is a sufficient condition for the existence of a pure strategy equilibrium.

\section{The Model}
In our model, there are three types of consumers, $C_1,C_2,C_3$, whose preferences are represented by the interval $[0,1]$. $C_1$'s ideal point of preference is $0$. $C_2$'s ideal point of preference is $1$. $C_i$ is a fan of Firm $i(=1,2)$. $C_3$ evaluates either firm's products. $C_3$'s ideal point of preference is at an equal distance from the fans of each firm, that is, at $1/2$ in $[0,1]$.
Now we assume symmetry, that is, an equal density of consumers exists at $0$ and at $1$. We normalize the density of consumers at $0$ and at $1$ as $1$ and set the density of consumers at $\frac{1}{2}$ as $x$. We assume $x \ge 1$. Each consumer $C_k\:(k=1,2,3)$ purchases one unit of the product from either one of the two firms.

Both firms choose their own location on interval $[0,1]$. $z_i$ denotes the location of firm $i$. $p_i$ denotes the price of the firm's product. We assume $0 \le z_1 \le 1/2 \le z_2 \le 1$. Each firm $i$ chooses its own location and, then, chooses a price for the product.

We normalize as $1$ the reservation value of a product that is sold by a firm patronized by $C_1$ and $C_2$. Then, we define $y$ as $C_3$'s reservation value for a product of firm $i\:(i=1,2)$. Consumer $C_i$ evaluates the product of firm $3-i$ at $0$. The distance between each consumer's ideal point and the location of a firm measures the consumer's disutility. In our model, this disutility is measured with the quadratic cost function.\footnote{In what follows, we sometimes call this ``transportation cost.''} 
%

$C_3$ prefers a product with a lower price including transportation costs. We define each consumer's utility when the consumer purchases a product with the characteristic $z_i$ at $p_i$ as follows. Every consumer chooses their behavior to maximize their utility.
\begin{align}\label{uc1}
\begin{cases}
u^{C_1}=1-\{p_1+z_1^2\},\\
u^{C_2}=1-\{p_2+(1-z_2)^2\},\\
u^{C_3}=y-\{p_i+(\frac{1}{2}-z_i)^2\}
\end{cases}
\end{align}

Now, we consider each consumer's choice and each firm's profit in a price subgame given the pair $(z_1,z_2)$. Each firm $i$'s profit is defined by the sum of profit gained from $C_i$ and $C_3$. Here, we assume that production cost is $0$. We consider Firm $1$'s profit. $C_2$ does not purchase a product from Firm $1$. Let $\pi_1(p_1,p_2)$ denote Firm $1$'s profit. $\pi_1(p_1,p_2)$ is the total sum of $\pi_1^{C_1}(p_1)$, which denotes Firm $1$'s profit gained from $C_1$ and $\pi_1^{C_3}(p_1,p_2)$, which denotes Firm $1$'s profit gained from $C_3$.
\begin{align}\label{pi1all}
\pi_1(p_1,p_2)=\pi_1^{C_1}(p_1)+\pi_1^{C_3}(p_1,p_2).
\end{align}

First, we define $\pi_1^{C_1}(p_1)$. By \eqref{uc1}, $C_1$ purchases Firm $1$'s product when $1-\{p_1+z_1^2\} \ge 0$ is satisfied. Therefore, we obtain $p_1 \le 1-z_1^2$. Now, we define Firm $1$'s profit obtained from $C_1$ as follows.
\begin{align}\label{pi1m1}
\pi_1^{C_1}(p_1)=
\begin{cases}
p_1, &\text{if $p_1 \le 1-z_1^2,$}\\
0, &\text{otherwise}
\end{cases}
\end{align}

Next, we define $\pi_1^{C_3}(p_1,p_2)$. We define $C_3$'s utility, $u^{C_3}$, as follows.
\begin{align*}
u^{C_3}=
\begin{cases}
y-\{p_1+(\frac{1}{2}-z_1)^2\}, &\text{if buy from 1,}\\
y-\{p_2+(\frac{1}{2}-z_2)^2\}, &\text{if buy from 2,}\\
0, &\text{otherwise.}
\end{cases}
\end{align*}
Here, we find that $C_3$ purchases only Firm $1$'s product when the following two equations are satisfied,
$p_1+(1/2-z_1)^2 < p_2+(1/2-z_2)^2, p_1+(1/2-z_1)^2 \le y$. It follows that
\begin{align*}
p_1 < p_2+(1/2-z_2)^2-(1/2-z_1)^2, p_1 \le y-(1/2-z_1)^2.
\end{align*}
Then, again, $C_3$ is divided equally between Firm $1$ and Firm $2$ when $C_3$ is indifferent between choosing either Firm $1$'s or Firm $2$'s product, $p_1+(1/2-z_1)^2 = p_2+(1/2-z_2)^2 \le y$.

Now, we define $\pi_1^{C_3}$, which denotes a firm's profit obtained from $C_3$, as follows.
\begin{align}\label{pi1m3}
\pi_1^{C_3}(p_1,p_2)=
\begin{cases}
p_1x, &\text{if $p_1 < p_2+(\frac{1}{2}-z_2)^2-(\frac{1}{2}-z_1)^2$ and $p_1 \le y-(\frac{1}{2}-z_1)^2,$}\\
p_1\frac{1}{2}x, &\text{if $p_1=p_2+(\frac{1}{2}-z_2)^2-(\frac{1}{2}-z_1)^2$ and $p_1 \le y-(\frac{1}{2}-z_1)^2,$}\\
0, &\text{otherwise}
\end{cases}
\end{align}

Substituting \eqref{pi1m1} and \eqref{pi1m3} with \eqref{pi1all}, we obtain Firm $1$'s profit in a price subgame given a pair of location points. We also define Firm $2$'s profit in the same way because of symmetry. Since Firm $2$'s profit $\pi_2$ is the total sum of profit gained from $C_2$ and $C_3$, we define $\pi_2$ as follows.
\begin{align}\label{pi2all}
\pi_2(p_1,p_2)=\pi_2^{C_2}(p_2)+\pi_2^{C_3}(p_1,p_2)
\end{align}

Now, we limit the range of $x,y$ as follows
\begin{align}\label{misutenai}
yx < \frac{3}{4}\left(1+x\right)
\end{align}
\begin{lem}\label{misutenai2}
When \eqref{misutenai} holds, in an equilibrium of a price subgame corresponding to a pair of location points, Firm $1$ does not choose a price such that $p_1>1-z_1^2$.
\end{lem}
\begin{proof}
Given $z_1,z_2$, we fix Firm $2$'s strategy in a price subgame. We must show that the total sum of profit gained from $C_1$ and $C_3$ when it charges the price $1-z_1^2$ is greater than the profit gained from $C_3$ when it charges the price $p_1>1-z_1^2$.

Let $\rho$ denote the probability of $C_3$ purchasing Firm $1$'s product when Firm $1$ charges $p_1$. Let $\rho^*$ denote the probability of $C_3$ purchasing Firm $1$'s product when Firm $1$ charges $1-z_1^2$.

Since $p_1>1-z_1^2$, we obtain $\rho \le \rho^*$. Here we show that $p_1 \rho x < (1-z_1^2)(1+\rho^*x)$.
When $\rho=0$, this equation is obvious. Suppose $\rho>0$.
Then, $p_1 \le y$ must hold. Now, by \eqref{misutenai}, we obtain
\begin{align*}
&(1-z_1^2)(1+\rho^* x)-p_1 \rho x \ge \rho(1-z_1^2)+[\rho^*(1-z_1^2)-\rho p_1]x \\
\ge&\rho[(1-z_1^2)(1+x)-p_1x] \ge \rho[(1-z_1^2)(1+x)-yx]>0
\end{align*}
\end{proof}
We have shown this lemma for Firm $1$. Because of symmetry, we can show that Firm $2$ does not charge a price such that $p_2 > 1-(1-z_2)^2$ in an equilibrium in a price subgame.

\section{Price Game}
In this section, we characterize an equilibrium profit in a price subgame corresponding to all pairs of $(z_1,z_2)$. Because of symmetry at $1/2$, without loss of generality, it is sufficient for us to focus on the case of $z_1+z_2 \le 1$. In the following, we focus on a mixed strategy equilibrium when both $z_1$ and $z_2$ are near to $1/2$. These include equilibria of Varian's model of price dispersion in which product differentiation is explicitly considered. However, according to a location pair that chooses their locations, a pure strategy equilibrium exists in a price subgame. See Appendix A for more details.

First, we characterize an equilibrium profit. When both $z_1$ and $z_2$ are close to $1/2$, the following equations are satisfied,
\begin{align}\label{tomonitikai}
\begin{cases}
1-z_1^2 < \{y-(1/2-z_1)^2\}(1+x),\\
1-(1-z_2)^2 < \{y-(1/2-z_2)^2\}(1+x).
\end{cases}
\end{align}

Let $(\pi^*_1,\pi_2^*)$ be an equilibrium profit vector, which is fixed in this price subgame.
Let $\bar{p}_i^*$ be the upper bound of firm $i$'s strategy in this equilibrium and $\underline{p}_i^*$ be
the lower bound of firm $i$'s strategy in this equilibrium.

We show $\pi_1^*=1-z_1^2$. Intuitively, initially, we find that Firm 1 does not charge $p_1>1-z_1^2$ because it might lose its own fan and, moreover, we show that Firm 1 sets $\underline{p}_1^* \ge (1-z_1^2)/(1+x)$. Finally, we determine that the equilibrium profit of Firm 1 is $\pi_1^*=1-z_1^2$. Once we determine the lower bound of the support of Firm 1's strategy and its equilibrium profit, we also determine Firm 2's lower bound $\underline{p}_2^*$.
We then obtain $\pi_2^*$. See Appendix B for more details.


Now, we show that the equilibrium profit of each firm is determined uniquely, $(\pi_1^*, \pi_2^*)=(1-z_1^2, {1-z_1^2}+\left[(1/2-z_1)^2-(1/2-z_2)^2\right](1+x))$. Here, we obtain the next proposition.
\begin{prop}\label{pricesubgameuniquneness}
\begin{align*}
(\pi_1^*, \pi_2^*)=(1-z_1^2, 1-z_1^2+[(1/2-z_1)^2-(1/2-z_2)^2](1+x)).
\end{align*}
\end{prop}
\begin{proof}
By Lemma~\ref{pi1uniquneness}-\ref{pi2upperbound}, we obtain this result. See Appendix B for more details.
\end{proof}

However, each firm cannot obtain this profit in a pure strategy equilibrium.
In the following, we show a mixed strategy equilibrium for firms to obtain this equilibrium profit.

\begin{prop}\label{toitoi1-z1^2+rent}
If
\begin{align*}
\begin{cases}
1-z_1^2 < \{y-(1/2-z_1)^2\}(1+x),\\
y-(1/2-z_2)^2 < 1-(1-z_2)^2
\end{cases}
\end{align*}
holds, then, an equilibrium exists such that the equilibrium strategy of each firm is
\begin{align}\label{toitoi1}
F^{1*}(p)=
\begin{cases}
0, &\text{if $p \le \frac{1-z_1^2}{1+x},$}\\
1-\frac{\pi_2^*-(p+T_1)}{(p+T_1)x}, &\text{if $p \in [\frac{1-z_1^2}{1+x},y-(1/2-z_1)^2],$}\\
F^{1*}(y-(1/2-z_1)^2), &\text{if $p \in [y-(1/2-z_1)^2,1-z_1^2),$}\\
1, &\text{if $p \ge 1-z_1^2.$}
\end{cases}
\end{align}
\begin{align}\label{toitoi2}
F^{2*}(p)=
\begin{cases}
0, &\text{if $p \le \frac{1-z_1^2}{1+x}+T_1,$}\\
1-\frac{\pi_1^*-(p-T_1)}{(p-T_1)x}, &\text{if $p \in [\frac{1-z_1^2}{1+x}+T_1,y-(1/2-z_2)^2),$}\\
1, &\text{if $p \ge y-(1/2-z_2)^2.$}
\end{cases}
\end{align}
where $T_1=(1/2-z_1)^2-(1/2-z_2)^2 \ge 0$. In this equilibrium, each firm's equilibrium profit is
$\pi_1^*=1-z_1^2, \pi_2^*=1-z_1^2+\{(1/2-z_1)^2-(1/2-z_2)^2\}(1+x)$.
\end{prop}
\begin{proof}
See Appendix C.
\end{proof}

\begin{remark}
Equilibrium strategies shown by Proposition~\ref{toitoi1-z1^2+rent} involve symmetric location pairs.
At symmetric location pairs, a symmetric equilibrium strategy profile exists where equilibrium profit is equal to equilibrium profit in an asymmetric equilibrium.
For example, an equilibrium exists such that
\begin{align*}
F^{1*}(p)=
\begin{cases}
0, &\text{if $p \le \frac{1-z_1^2}{1+x},$}\\
1-\frac{\pi_2^*-p}{px}, &\text{if $p \in [\frac{1-z_1^2}{1+x},y-(1/2-z_1)^2],$}\\
F^{1*}(y-(1/2-z_1)^2), &\text{if $p \in [y-(1/2-z_1)^2,1-z_1^2),$}\\
1, &\text{if $p \ge 1-z_1^2.$}
\end{cases}
\end{align*}
\begin{align*}
F^{2*}(p)=
\begin{cases}
0, &\text{if $p \le \frac{1-(1-z_2)^2}{1+x},$}\\
1-\frac{\pi_1^*-p}{px}, &\text{if $p \in [\frac{1-(1-z_2)^2}{1+x},y-(1/2-z_2)^2),$}\\
F^{2*}(y-(1/2-z_2)^2), &\text{if $p \in [y-(1/2-z_2)^2,1-(1-z_2)^2),$}\\
1, &\text{if $p \ge 1-(1-z_2)^2.$}
\end{cases}
\end{align*}
\end{remark}

\begin{prop}\label{1toi2tikaimixed1-z1^2+rent}
If
\begin{align*}
\begin{cases}
1-z_1^2 < \{y-(1/2-z_1)^2\}(1+x),\\
1-(1-z_2)^2 \le y-(1/2-z_2)^2
\end{cases}
\end{align*}
holds, then an equilibrium exists such that the equilibrium strategy of each firm is
\begin{align*}
F^{1*}(p)=
\begin{cases}
0, &\text{if $p \le \frac{1-z_1^2}{1+x},$}\\
1-\frac{\pi_2^*-(p+T_1)}{(p+T_1)x}, &\text{if $p \in [\frac{1-z_1^2}{1+x},1-(1-z_2)^2-T_1],$}\\
F(1-(1-z_2)^2-T_1), &\text{if $p \in [1-(1-z_2)^2-T_1,1-z_1^2),$}\\
1, &\text{if $p \ge 1-z_1^2.$}
\end{cases}
\end{align*}
\begin{align*}
F^{2*}(p)=
\begin{cases}
0, &\text{if $p \le \frac{1-z_1^2}{1+x}+T_1,$}\\
1-\frac{\pi_1^*-(p-T_1)}{(p-T_1)x}, &\text{if $p \in [\frac{1-z_1^2}{1+x}+T_1,1-(1-z_2)^2),$}\\
1, &\text{if $p_2 \ge 1-(1-z_2)^2$}
\end{cases}
\end{align*}
where $T_1=(1/2-z_1)^2-(1/2-z_2)^2 \ge 0.$ In this equilibrium, each firm's equilibrium profit is $\pi_1^*=1-z_1^2$, and $\pi_2^*=1-z_1^2+\{(1/2-z_1)^2-(1/2-z_2)^2\}(1+x)$.
\end{prop}
\begin{proof}
In an equilibrium, $\pi_i^*, (i=1,2)$ is constant regardless of the selected strategy. This is the same as we have shown for Proposition~\ref{toitoi1-z1^2+rent}.
We show that strategy profiles are an equilibrium. By $\pi_1^* = 1-z_1^2$, Firm $1$ does not deviate to a price strictly lower than $\frac{1-z_1^2}{1+x}$. Because Firm $1$'s lower bound of support for its equilibrium strategy is $\frac{1-z_1^2}{1+x}$, it does not follow that Firm $2$ charges a price strictly lower than $\frac{1-z_1^2}{1+x}+T_1$, which is the lower bound of its strategy support.

Next, we show that Firm $1$ does not deviate, choosing a price $\dot{p}_1$ such that $1-(1-z_1)^2-T_1 < \dot{p}_1 < 1-z_1^2$.
Supposing that $F^{2*}(p)$ is given, it follows with a probability of $1$ that Firm $2$'s price charged to $C_3$ including transport costs is strictly lower than Firm $1$'s price. Thus, only $C_1$ purchases Firm $1$'s product. On the other hand, in this case, Firm $1$'s profit is strictly lower than $\pi_1^*$. Thus, Firm $1$ does not deviate.
\end{proof}

\begin{lem}\label{pi2starcondition}
\begin{align*}
&1-z_1^2+[(1/2-z_1)^2-(1/2-z_2)^2](1+x) > 1-(1-z_2)^2, \quad\text{if $z_1+z_2 < 1,$}\\
&1-z_1^2+[(1/2-z_1)^2-(1/2-z_2)^2](1+x) = 1-(1-z_2)^2, \quad\text{if $z_1+z_2 = 1.$}
\end{align*}
\end{lem}
\begin{proof}
We obtain the result immediately from \eqref{neweq}
\end{proof}

\begin{lem}\label{pi2starcondition2}
If $1-z_1^2 < \{y-(1/2-z_1)^2\}(1+x)$ is satisfied, then
\begin{align*}
\{y-(1/2-z_2)^2\}(1+x) > 1-z_1^2+[(1/2-z_1)^2-(1/2-z_2)^2](1+x).
\end{align*}
\end{lem}
\begin{proof}
Given that condition $1-z_1^2 < \{y-(1/2-z_1)^2\}(1+x)$ holds, we determine that
$1-z_1^2 + [(1/2-z_1)^2-(1/2-z_2)^2](1+x) < \{y-(1/2-z_1)^2\}(1+x) +[(1/2-z_1)^2-(1/2-z_2)^2](1+x) = \{y-(1/2-z_2)^2\}(1+x).$
\end{proof}

\begin{remark}\label{remark4}
We consider an implication of $T$ as follows. In a price subgame where $z_1+z_2 \le 1$ holds, Firm $2$ is located close to $C_3$ while Firm $1$ is located close to its monopoly market. Firm $1$ gains at most its monopoly profit because it is located further away from $C_3$ than Firm $2$.

On the other hand, by $z_1+z_2<1$, Firm $2$ gains greater profit than Firm $1$ by $\{(1/2-z_1)^2-(1/2-z_2)^2\}(1+x) > 0$ because Firm $2$ is located closer to $C_3$ than Firm $1$.

Moreover, by Lemma~\ref{pi2starcondition}, Firm $2$'s profit fails to achieve the maximum level it could achieve at its present location if it plays as a monopolist at this location. 
However, by Lemma~\ref{pi2starcondition2}, Firm $2$'s profit fails to achieve the maximum profit that it gains at its present location if it plays as a monopolist at this location.

This is because the competitive effect of Firm $1$'s behavior entails a rent that Firm $2$ enjoys.
\end{remark}

\section{Location Game}

Let $\Pi_i(z_1,z_2)$ denote firm $i$'s profit in a location game.
\begin{align}\label{locpayoff}
\Pi_1(z_1,z_2)=
\begin{cases}
1-z_1^2, &\text{if $z_1 \in [0, \bar{z}_1],$}\\
(y-(1/2-z_1)^2)(1+x), &\text{if $z_1 \in [\bar{z}_1,1/2], z_2 \in [\bar{z}_2,1],$}\\
1-z_1^2, &\text{if $z_1 \in [\bar{z}_1,1/2], z_2 \in [1/2,\bar{z}_2], z_1+z_2 \le 1,$}\\
1-(1-z_2)^2+T_2(1+x), &\text{if $z_1 \in [\bar{z}_1,1/2], z_2 \in [1/2,\bar{z}_2], z_1+z_2 > 1,$}
\end{cases}
\end{align}

\begin{align}\label{locpayoff2}
\Pi_2(z_1,z_2)=
\begin{cases}
1-(1-z_2)^2, &\text{if $z_2 \in [\bar{z}_2,1],$}\\
(y-(1/2-z_2)^2)(1+x), &\text{if $z_1 \in [0,\bar{z}_1], z_2 \in [1/2,\bar{z}_2],$}\\
1-(1-z_2)^2, &\text{if $z_1 \in [\bar{z}_1,1/2], z_2 \in [1/2,\bar{z}_2], z_1+z_2 \ge 1,$}\\
1-z_1^2+T_1(1+x), &\text{if $z_1 \in [\bar{z}_1,1/2], z_2 \in [1/2,\bar{z}_2], z_1+z_2 < 1.$}
\end{cases}
\end{align}
, where $T_j=(1/2-z_j)^2-(1/2-z_i)^2, (i=1,2, j \ne i)$.

In this section, we consider symmetric location equilibrium. First, we show that equilibrium pay-offs in all equilibrium location pairs are identical, that is, equal to $1$.
\begin{prop}\label{overallgamepayoff}
Equilibrium pay-offs for any subgame perfect equilibrium of this game are equal to $1$.
\end{prop}
\begin{proof}
Firm $1$ gains a profit $\Pi_1=1$ with a probability of $1$ if it is located at $z_1=0$. It follows that equilibrium profit is not strictly less than $1$. Suppose that Firm $1$ gains a profit strictly greater than $1$. If Firm $1$ is located at $\tilde{z}_1$, which is the left edge of its location strategy, it follows from the discussion above that its profit is $1-\tilde{z}_1^2$. However, this profit is less than $1$.
Because of continuity in profit functions in all location points, Firm $1$ cannot gain an equilibrium profit greater than $1$ that is approximately $\tilde{z}_1$. Thus, Firm $1$ cannot gain a profit strictly greater than $1$ in an equilibrium.
\end{proof}

In the following, we consider the case of $y(1+x) \ge 1$. We show that a symmetric mixed strategy equilibrium exists in this case.
\paragraph{Mixed Strategy Equilibrium $1$}
We consider equilibrium location pairs when $x,y$ is such that $(y-1/4)(1+x) > 1$.
It follows by \eqref{tomonitikai} that $\hat{z}_1, \hat{z}_2$ does not exist.
Thus, we obtain pay-off functions as follows,
\begin{align}\label{eachprofitatlocation1}
\Pi_1(z_1,z_2)=
\begin{cases}
1-z_1^2, &\text{if $z_1+z_2 \le 1,$}\\
1-(1-z_2)^2+T_2(1+x), &\text{if $z_1+z_2 > 1.$}
\end{cases}
\end{align}
\begin{align}\label{eachprofitatlocation2}
\Pi_2(z_1,z_2)=
\begin{cases}
1-(1-z_2)^2, &\text{if $z_1+z_2 \ge 1,$}\\
1-z_1^2+T_1(1+x), &\text{if $z_1+z_2 < 1.$}
\end{cases}
\end{align}
The closer a firm is located to $1/2$, the more profit it gains when it is located asymmetrically with respect to the other firm.
However, both firms do obtain at most their own monopoly profits at symmetric location pairs. Thus, since $\left(\frac{1}{2}, \frac{1}{2}\right)$ is not a best response pair, it follows that an equilibrium location pair does not exist using a pure strategy. In the following, we present an equilibrium location pair using a mixed strategy.

Given Firm 2's strategy $G_2(z_2)$, the expected pay-off of Firm 1 when located at $z_1$ is given by the following:
\begin{align}\label{epi1}
E[\Pi_1]=&\int_{1/2}^{1-z_1}(1-z_1^2)g_2(z_2)dz_2\notag\\
&+\int_{1-z_1}^1\left[1-(1-z_2)^2+\{(1/2-z_2)^2-(1/2-z_1)^2\}(1+x)\right]g_2(z_2)dz_2,
\end{align}
where $g_2(z_2)$ denotes the density function of a distribution function $G_2$.

\begin{prop}\label{mixmix}
Let $x,y$ satisfy $(y-1/4)(1+x) > 1$. Then, a mixed strategy equilibrium exists such that the equilibrium strategy of each firm is
\begin{align*}
&G_1^*(z_1)=
\dfrac{2z_1}{x(1-2z_1)+1},&\quad\text{$0 \le z_1 \le 1/2,$}\\
&G_2^*(z_2)=
\dfrac{(2z_2-1)(1+x)}{x(2z_2-1)+1},&\quad\text{$1/2 \le z_2 \le 1,$}\\
\end{align*}
where $G_2^*(z_2)$ is a distribution function because $G{_2^*}'>0, G_2^*(\frac{1}{2})=0$, and $G_2^*(1)=1$. $G_1^*(z_1)$ is a distribution function for similar reasons.
\end{prop}
\begin{proof}
We show that for any given $G_2^*(z_2)$, Firm $1$ gains an expected profit equal to $1$ when it is located at $z_1 \in [0,1/2]$.
Substituting $G_2^*$ into \eqref{epi1}, by partial integration for the second term of this equation, we obtain
\begin{align*}
&E[\Pi_1^*]=(1-z_1^2)[G_2^*(z_2)]^{1-z_1}_{1/2}\notag\\
&+\left[[1-(1-z_2)^2+\{(1/2-z_2)^2-(1/2-z_1)^2\}(1+x)]G_2^*(z_2)\right]^1_{1-z_1}\notag\\
&-\int^{1}_{1-z_1}[1-(1-z_2)^2+\{(1/2-z_2)^2-(1/2-z_1)^2\}(1+x)]^{\prime}G_2^*(z_2)dz_2\notag\\
&=1+{(1/4-(1/2-z_1)^2)(1+x)}-(1+x)\int^{1}_{1-z_1}(2z_2-1)dz_2\notag\\
&=1+{(1/4-(1/2-z_1)^2)(1+x)}-(1+x)\left[\left(z_2-\frac{1}{2}\right)^2\right]^{1}_{1-z_1}=1.
\end{align*}
By symmetry, we obtain the same result for Firm $2$ in the same way.
\end{proof}

\paragraph{Mixed Strategy Equilibrium $2$}
Next, we consider the case where $x,y$ lies in the region where both $y(1+x) \ge 1$ and $(y-1/4)(1+x) \le 1$ hold.
This region is characterized by $\bar{z}$. \eqref{locpayoff}-\eqref{locpayoff2} denotes the respective pay-offs for Firms 1 and 2. In this case, we construct an equilibrium using $G^*$ obtained through Proposition~\ref{mixmix}.
Now, we consider $\hat{z}_2$, which satisfies
\begin{align}\label{hatz2def}
(1-(1-\hat{z}_2)^2)G_2^*(\hat{z}_2)+\{1-G_2^*(\hat{z}_2)\}\left\{y-(\hat{z}_2-1/2)^2\right\}(1+x)=1
\end{align}

Solving (6), we obtain the following:
\begin{align*}
\hat{z}_2=\frac{1}{2}+2\left(y-\frac{1}{1+x}\right).
\end{align*}
Now, from $y(1+x) \ge 1$, it follows that $y-\frac{1}{1+x} \ge 0$ and, then, it follows that $1/2 \le \hat{z}_2$.

Next, we show $\hat{z}_2 \le \bar{z}_2$. From the definition of $\bar{z}_2$, we obtain $(y-(1/2-\bar{z}_2)^2)(1+x)=1-(1-\bar{z}_2)^2$.
The left-hand side of this equation is monotone decreasing with respect to $\bar{z}_2$ while the right-hand side is monotone increasing.
Thus, in the following, we show that $(y-(1/2-\hat{z}_2)^2)(1+x) \ge 1-(1-\hat{z}_2)^2$.

By $(y-(\hat{z}_2-1/2)^2)(1+x)=1+\frac{1}{2}(1-\hat{z}_2)(2\hat{z}_2-1)(1+x)$, we determine that $\frac{1}{2}(1-\hat{z}_2)(2\hat{z}_2-1)(1+x) > -(1-\hat{z}_2)^2$ and, thus, $\hat{z}_2 \le \bar{z}_2$.

In the same way, by symmetry, we show that $\hat{z}_1=\frac{1}{2}-2\left(y-\frac{1}{1+x}\right).$

Using $\hat{z}_2$, we show that a symmetric equilibrium exists such that Firm $2$ plays a mixed strategy, which is the same as $G_2^*$ on the support $[1/2,\hat{z}_2]$ while the remainder of the probability, $1-G_2^*$, is attached to being located at $z_2=1$.

Given this strategy of Firm $2$, $E[\Pi_1]$ denotes the expected profit of Firm $1$ for all points $z_1$ on the interval $[1-\hat{z}_2,1/2]$. We conclude that
\begin{align}\label{firm1EXPECTEDPAYOFF}
E[\Pi_1]=&\int^{1-z_1}_{1/2}(1-z_1^2)g_2(z_2)dz_2\notag\\
&+\int^{\hat{z}_2}_{1-z_1}[1-(1-z_2)^2+\left\{(1/2-z_2)^2-(1/2-z_1)^2\right\}(1+x)]g_2(z_2)dz_2\notag\\
&+\left\{1-G_2(\hat{z}_2)\right\}\left\{y-(1/2-z_1)^2\right\}(1+x),
\end{align}
where $g_2(z_2)$ denotes the density function of the distribution function $G_2$.

\begin{prop}\label{mixmixalpha1}
Let $x,y$ satisfy both $y(1+x) \ge 1$ and $(y-1/4)(1+x) \le 1$. Then, an equilibrium exists such that each firm plays.
\begin{align*}
&G_1^{**}(z_1)=
\begin{cases}
0,&\quad\text{if $z_1 < 0,$}\\
\dfrac{2\hat{z}_1}{x(1-2\hat{z}_1)+1},&\quad\text{if $0 \le z_1 \le \hat{z}_1,$}\\
\dfrac{2z_1}{x(1-2z_1)+1},&\quad\text{if $\hat{z}_1 \le z_1 \le 1/2,$}\\
\end{cases}\\
&G_2^{**}(z_2)=
\begin{cases}
\dfrac{(2z_2-1)(1+x)}{x(2z_2-1)+1},&\quad\text{if $1/2 \le z_2 \le \hat{z}_2,$}\\
\dfrac{(2\hat{z}_2-1)(1+x)}{x(2\hat{z}_2-1)+1},&\quad\text{if $\hat{z}_2 \le z_2 < 1,$}\\
1,&\quad\text{if $1 \le z_2.$}
\end{cases}
\end{align*}
\end{prop}
\begin{proof}
See Appendix C.
\end{proof}

\paragraph{Pure Strategy Equilibrium}
When $y$ is relatively small, $C_3$ consumers do not want to purchase any products because they do not feel sufficiently attracted to the products.
This is the reason firms sell products that have a distinctive characteristic recognized by loyal consumers.
This situation can occur when $C_3$ is a large market but the reservation value $y$ is sufficiently small.
Thus, it appears that $C_3$ does not exist.
Finally, we characterize this type of subgame as a perfect equilibrium.

\begin{prop}\label{winloseloc}
When $x,y$ satisfies $1 > y(1+x)$, we have a subgame perfect equilibrium such that
\begin{align*}
z_1^*=0, z_2^*=1.
\end{align*}
\end{prop}
\begin{proof}
Firm $1$ gains a profit $\max\{1-z_1^2,(y-(1/2-z_1)^2)(1+x)\}$ when it is located at $z_1$.
Given $z_1^*=0,z_2^*=1$, suppose that Firm $1$ deviates to $z_1 \in (0,1/2]$.
Because of $1 > y(1+x)$, Firm $1$ gains a profit strictly less than $1$.
Therefore, Firm $1$ never deviates. Because of symmetry, we show this is the same for Firm $2$.
Thus, $z_1^*=0, z_2^*=1$ is a subgame perfect location pair.
\end{proof}

\begin{remark}
If $y(1+x) \ge 1$ is satisfied, an asymmetric pure strategy equilibrium exists such that $(z_1^*,z_2^*)=\left(0,\frac{1}{2}\right) \text{or} \left(\frac{1}{2},1\right)$. This corresponds to the region mentioned in Proposition~\ref{mixmix} and \ref{mixmixalpha1}. If $(y-1/4)(1+x) \le 1$ is satisfied, firms behave in the same manner as they would in a monopoly. That is, a firm located at $0$ or $1$ maximizes its own monopoly profit while a firm located at $1/2$ maximizes the profit that it gains from both its own fans and $C_3$ consumers. Thus, the firms do not deviate. If $(y-1/4)(1+x) > 1$ is satisfied, because a firm located at either $0$ or $1$ is always far away from $1/2$ whichever way it moves, it cannot gain a strictly higher profit than monopoly profit. Thus, the firm does not deviate. On the other hand, given the strategy of a firm that specializes in its own fans, a firm located at $1/2$ loses its rent, mentioned in Remark~\ref{remark4}, when it moves from $1/2$ to another location. Thus, this firm does not deviate.
\end{remark}

\section{Conclusions and Remarks}
We consider a model that incorporates a spatial competition framework \`a la Hotelling (1929) where firms compete in price and location in Varian's model of sales.

We show that the difference between each firm's choice of location in a subgame perfect equilibrium will be neither purely maximized at both ends of a line nor purely minimized at the center. In our model, these two typical results stochastically occur in a subgame perfect equilibrium.

Here, we consider why firms' choice of locations is randomized. Discontinuous games are analyzed by a mixed strategy equilibrium. However, in our model, a pay-off function for location points $z$ is continuous at the point where both firms' products are of equal value for the informed $C_3$. Thus, we find that ``dispersion'' is not caused by discontinuity in profit functions with regard to $z$.

On the other hand, Varian's model of sales is also characterized by a pay-off discontinuity, which is caused by discontinuous demands. In our model, each firm has its own fans, that is, the firms shield a part of their own markets while engaging in a Bertrand competition over the informed who are located at 1/2 in the interval. At the same time, in this market, firms cannot set two-part prices, one for their fans and one to attract the informed switchers. This structure is the same as the relationship between informed and uninformed consumers in Varian's model. This is the reason a pay-off discontinuity occurs in the profits arising in the pricing subgame of our model. This discontinuity in profits causes price dispersion in a price subgame, just as sales arise in Varian's model. This is the reason why Varian's similar result holds even when firm location is explicitly included.

However, for the first-stage game truncated by backward induction, a pay-off function for characteristic Z is continuous at the point where both firms' products are of equal value to the switchers. Therefore, the occurrence of randomization in a choice of product characteristics is not caused by discontinuity in profit functions for the choice of characteristics. Thus we find that the existence of ``fans'' is key to the cause of ``dispersion'' of product differentiation. This result is caused by the explicit introduction of choice in characteristics into Varian's model of sales.

Moreover, unlike Varian's model, we show that by explicitly introducing location choice, it is possible for a firm to gain excess profit in a price competitive market. However, from Proposition 4, firms' expected profit in a subgame perfect equilibrium is set as the maximum monopoly profit, equal to one. Thus, even when product differentiation is explicitly introduced into a Varian-type model, Varian's implication can be retained; the opportunity for profit in an informed market is lost with competition.

\appendix
\section{Pure strategy equilibriums}
In the following, we consider pure strategy equilibriums. In each case, we show that an equilibrium profit in a subgame is unique.

First, these cases are classified according to whether $z_1$ is close to $1/2$ or far from $1/2$. When $z_1$ is far from $1/2$, Firm $1$ gains a higher profit when it sells a product to $C_1$ only compared to when it sells to both $C_1$ and $C_3$. We consider the case of $1-z_1^2 \ge \{y-(1/2-z_1)^2\}(1+x)$. When $z_1$ is close to $1/2$, the definition is inverted.

Now, we define $\bar{z}_1$, which is given by evaluating $(y-(1/2-z_1)^2)(1+x)=1-z_1^2$ for $z_1$. By $\bar{z}_1 \in [0,1/2]$,
\begin{align*}
\bar{z}_1 = \frac{1}{2}-\frac{\sqrt{(x-1)^2+4x(y-1/4)(1+x)}-1}{2x}.
\end{align*}
\begin{remark}\label{barznonexist}
For all $z_1$, $(y-(1/2-z_1)^2)(1+x) > 1-z_1^2$ if and only if $(y-1/4)(1+x) > 1$. Therefore, $\bar{z}_1$ does not exist if $(y-1/4)(1+x) > 1$.
Now, we assume that $x \ge 1$. Moreover, it follows that if $\bar{z}_1$ exists, then, $y \le 3/4$.
\end{remark}
Because of symmetry, we define $\bar{z}_2$ in the same way.
\begin{align*}
\bar{z}_2 = \frac{1}{2}+\frac{\sqrt{(x-1)^2+4x(y-1/4)(1+x)}-1}{2x}.
\end{align*}

Here, we fix case $z_1 \in [0,\bar{z}_1]$ where $z_1$ is far from $1/2$. We classify the following two cases according to the distance between $z_2$ and $1/2$. At first, $z_2$ is also far from $1/2$, that is, $z_2 \in [\bar{z}_2,1]$. Thus, $z_2$ satisfies $1-(1-z_2)^2 \ge \{y-(1/2-z_2)^2\}(1+x)$ (Proposition~\ref{pure1}). Next, only $z_2$ is close to $1/2$, that is, $z_2 \in [1/2,\bar{z}_2)$. Therefore, $z_2$ satisfies $1-(1-z_2)^2 < \{y-(1/2-z_2)^2\}(1+x)$ (Proposition~\ref{pure3}).

When for all $z_1 \in (\bar{z}_1,1/2]$, $1-z_1^2 < \{y-(1/2-z_1)^2\}(1+x)$ is satisfied, $z_1$ is close to $1/2$. Given the condition $z_1+z_2 \le 1$ mentioned above, we focus on the case where $z_2$ is also close to $1/2$, that is, $z_2 \in [1/2,\bar{z}_2)$, which follows from $1-(1-z_2)^2 < \{y-(1/2-z_2)^2\}(1+x)$. We have already considered this case. (See also Proposition~\ref{pricesubgameuniquneness} and Appendix B.)

\begin{prop}\label{pure1}
If
\begin{align}
1-z_1^2 &\ge \{y-(1/2-z_1)^2\}(1+x),\label{hikikomorisubgame1}\\
1-(1-z_2)^2 &\ge \{y-(1/2-z_2)^2\}(1+x)\label{hikikomorisubgame2}
\end{align}
holds, then the equilibrium profit vector is $(1-z_1^2, 1-(1-z_2)^2)$, and its profit is unique.
\end{prop}
\begin{proof}
Firm $1$'s maximized profit gained from $C_1$ only is higher than the maximum profit gained from both $C_1$ and $C_3$.

Thus, we obtain $\max_{p_1,p_2}\pi_1(p_1,p_2)=1-z_1^2$. Therefore, in an equilibrium, Firm $1$ never obtains a strictly higher profit than $1-z_1^2$.

Next, Firm $1$ obtains $1-z_1^2$ whenever it charges $p_1=1-z_1^2$ for any $p_2$. Thus, in an equilibrium, Firm $1$ never obtains strictly less than $1-z_1^2$. Similarly, Firm $2$ never obtains a profit other than $1-(1-z_2)^2$ in an equilibrium. Here, price profile $(1-z_1^2, 1-(1-z_2)^2)$ enables each firm to obtain maximum profit. Thus, this profile is a unique equilibrium.
\end{proof}

\begin{remark}
Proposition~\ref{pure1} shows that an equilibrium profit is unique in a price subgame. However, there may be more than one equilibrium, that is, an equilibrium strategy that achieves an equilibrium profit is not always unique. For example, when
\begin{align}
1-z_1^2 = \{y-(1/2-z_1)^2\}(1+x),\label{hikikomorisubgameremark1}\\
1-(1-z_2)^2 > \{y-(1/2-z_2)^2\}(1+x) \label{hikikomorisubgameremark2}
\end{align}
hold. It follows by \eqref{hikikomorisubgameremark2} that a price vector $1-(1-z_2)^2$ is a dominant strategy of Firm $2$. On the other hand, if Firm $2$ chooses $1-(1-z_2)^2$, it follows by \eqref{hikikomorisubgameremark1} that Firm $1$ is indifferent to both $1-z_1^2$ and $y-(1/2-z_1)^2$. Thus, we find that all the strategy pairs are equilibria where Firm $1$ uses a mixed strategy that combines $1-z_1^2$ with $y-(1/2-z_1)^2$ at any ratio it likes if Firm $2$ chooses $1-(1-z_2)^2$.
\end{remark}

\begin{prop}\label{pure3}
If
\begin{align}
1-z_1^2 &\ge \{y-(1/2-z_1)^2\}(1+x),\label{pure3p1}\\
1-(1-z_2)^2 &< \{y-(1/2-z_2)^2\}(1+x) \label{pure3p2}
\end{align}
hold. Then, the equilibrium profit is unique and this equilibrium profit is $(1-z_1^2, \{y-(1/2-z_2)^2\}(1+x))$
\end{prop}
\begin{proof}
We show Firm $1$'s profit. By \eqref{pure3p1}, we obtain $\max_{p_1,p_2}\pi_1(p_1,p_2)=1-z_1^2\ge\{y-(1/2-z_1)^2\}(1+x)$. Thus, in an equilibrium, Firm $1$ never obtains a strictly higher profit than $1-z_1^2$. Moreover, Firm $1$ obtains $1-z_1^2$ whenever it charges $p_1=1-z_1^2$ for any $p_2$. Thus, in an equilibrium, Firm $1$ never obtains strictly less than $1-z_1^2$.

In an equilibrium, Firm $1$ never chooses a price other than $1-z_1^2$ and $y-(1/2-z_1)^2$. Next we show that Firm $1$ never chooses $y-(1/2-z_1)^2$ at any positive probability in an equilibrium. If Firm $1$ plays a mixed strategy profile in which it chooses either $y-(1/2-z_1)^2$ at a positive probability $\eta > 0$ or $1-z_1^2$ at a positive probability $1-\eta > 0$, Firm $2$ does not have a best response to the strategy of Firm $1$. This is because by \eqref{pure3p2}, if Firm $2$ chooses a price $y-(1/2-z_2)^2-\varepsilon, (\varepsilon > 0)$ because of Remark~\ref{barznonexist}, it obtains both $C_2$ and $C_3$ and gains a profit $\{y-(1/2-z_2)^2-\varepsilon\}(1+x)$. However, if Firm $2$ chooses $y-(1/2-z_2)^2$, it obtains at most $\{y-(1/2-z_2)^2\}(1+x/2)$ because the prices of both firms, including transportation costs, are equal for $C_3$.

Thus, in an equilibrium, Firm $1$ chooses a price of $1-z_1^2$ at probability $1$. Given this strategy of Firm $1$, by \eqref{pure3p2}, Firm $2$'s best response is uniquely determined. Therefore, a unique equilibrium $(1-z_1^2,y-(1/2-z_2)^2)$ exists. Thus, a unique equilibrium profit $(1-z_1^2, \{y-(1/2-z_2)^2\}(1+x))$ is determined.
\end{proof}

\section{Uniqueness of equilibrium profit}
\begin{lem}\label{ikaarienai}
$\pi_1^* \ge 1-z_1^2$.
\end{lem}
\begin{proof}
Firm $1$ obtains a profit $1-z_1^2$ whenever it chooses $p_1=1-z_1^2$ for any of Firm $2$'s strategies. If $\pi_1^* < 1-z_1^2$, Firm $1$ increases its profit to $1-z_1^2$. Thus, this is contradictory to an equilibrium.
\end{proof}

\begin{lem}\label{gennmituniue}
If $\pi_1^*>1-z_1^2$ then $\underline{p}_1^* > \frac{1-z_1^2}{1+x}.$
\end{lem}
\begin{proof}
If $\pi_1^*>1-z_1^2$ and $\underline{p}_1^* \le \frac{1-z_1^2}{1+x}$, Firm $1$ obtains at most $1-z_1^2+(1+x)\varepsilon$ when it chooses $\underline{p}_1^*+\varepsilon, (\varepsilon>0)$. This profit is strictly less than $\pi_1^*$ when $\varepsilon \rightarrow 0$.
Therefore, this is contradictory to a definition of $\underline{p}_1^*$, which is the lower bound of support of an equilibrium strategy.
\end{proof}

\begin{lem}\label{pi1uniquneness}
\begin{align*}
\pi_1^*=1-z_1^2.
\end{align*}
\end{lem}
\begin{proof}
If $\pi_1^* > 1-z_1^2$ holds, then, it follows by \eqref{tomonitikai} and Lemma\ref{gennmituniue} that a positive number $\varepsilon$ exists such that
\begin{align}
&\underline{p}_1^*-\varepsilon > \frac{1-z_1^2}{1+x}, \label{check1}\\
&\frac{1-(1-z_2)^2}{1+x}+\varepsilon < \min\{1-(1-z_2)^2,y-(1/2-z_2)^2\}. \label{check2}
\end{align}
holds.

Now, we consider Firm $2$'s profit when it chooses
\begin{align*}
\tilde{p}_2 = \frac{1-(1-z_2)^2}{1+x}+\varepsilon
\end{align*}
for Firm $1$'s strategy. By \eqref{check1},
\begin{align}
&\underline{p}_1^*+(1/2-z_1)^2-[\tilde{p}_2+(1/2-z_2)^2]\notag\\
&=\underline{p}_1^*-\varepsilon-\frac{1-(1-z_2)^2}{1+x}+(1/2-z_1)^2-(1/2-z_2)^2\notag\\
&>\frac{1-z_1^2}{1+x}-\frac{1-(1-z_2)^2}{1+x}+(1/2-z_1)^2-(1/2-z_2)^2 \ge 0 \label{neweq}
\end{align}
are satisfied. Thus, we have $\underline{p}_1^*+(1/2-z_1)^2>\tilde{p}_2+(1/2-z_2)^2$. Additionally, by \eqref{check2}, we find that Firm $2$ gains a profit $\tilde{p}_2(1+x)=1-(1-z_2)^2+\varepsilon(1+x)$ from both $C_2$ and $C_3$. Thus, we have $\pi_2^* \ge \tilde{p}_2(1+x) > 1-(1-z_2)^2$. Therefore, we obtain $\pi_2^* > 1-(1-z_2)^2$, if $\pi_1^* > 1-z_1^2$.

Now, by Lemma \ref{misutenai2}, we have
\begin{align*}
\bar{p}_1^* \le 1-z_1^2.
\end{align*}
Therefore, it suffices to show that $\bar{p}_2^* \ge \bar{p}_1^*+(1/2-z_1)^2-(1/2-z_2)^2$. If $\bar{p}_2^* < \bar{p}_1^*+(1/2-z_1)^2-(1/2-z_2)^2$, Firm $1$ obtains at most $1-z_1^2$ when it chooses $\bar{p}_1^*$ because it never obtains $C_3$. This is contradictory to $\pi^*_1>1-z_1^2$.

Here, we find that $\pi_2^*>1-(1-z_2)^2$ is inconsistent with $\bar{p}_2^* > \bar{p}_1^*+(1/2-z_1)^2-(1/2-z_2)^2$. This is because Firm $2$ obtains at most $1-(1-z_2)^2$ when it chooses $\bar{p}_2^*$. Thus, we obtain $\bar{p}_2^* = \bar{p}_1^*+(1/2-z_1)^2-(1/2-z_2)^2$.

Assuming that Firm $1$ does not choose $\bar{p}_1^*$ with a positive probability, Firm $2$ obtains at most $1-(1-z_2)^2$ when it chooses $\bar{p}_2^*$. This is contradictory to $\pi_2^*>1-(1-z_2)^2$. In the same way, if Firm $2$ does not choose $\bar{p}_2^*$ with a positive probability, Firm $1$ obtains at most $1-z_1^2$ when it chooses $\bar{p}_1^*$. This is contradictory to $\pi_1^*>1-z_1^2$.
Thus, both firms choose $\bar{p}_i^*$ with a positive probability.

Let $F^i(p)$ be a probability distribution function when firm $i$ chooses a price less than $p$.
We define Firm $1$'s profit as follows when Firm $1$ chooses $\bar{p}_1^*$,
\begin{align}\label{uptie}
\bar{p}_1^*\left[1+\left\{1-\lim_{\varepsilon \rightarrow 0}F^2(\bar{p}^*_2-\varepsilon)\right\}\frac{x}{2}\right].
\end{align}
If Firm $1$ chooses $\bar{p}_1^*-\varepsilon$, it obtains at least
\begin{align}\label{downnige}
(\bar{p}_1^*-\varepsilon)\left[1+\left\{1-\lim_{\varepsilon \rightarrow 0}F^2(\bar{p}_2^*-\varepsilon)\right\}x\right].
\end{align}
If $\varepsilon \rightarrow 0$, then \eqref{uptie} $<$ \eqref{downnige}. This is contradictory to an optimality that Firm $1$ chooses $\bar{p}_1^*$ with a positive probability. Thus, we obtain $\pi_1^* = 1-z_1^2$.
\end{proof}

\begin{lem}\label{pi2upperbound}
\begin{align*}
\pi_2^* = 1-z_1^2+[(1/2-z_1)^2-(1/2-z_2)^2](1+x).
\end{align*}
\end{lem}
\begin{proof}
By Lemma~\ref{pi1uniquneness}, $\pi_1^* =1-z_1^2$. We obtain
\begin{align*}
\underline{p}_1^* \ge \frac{1-z_1^2}{1+x}.
\end{align*}

Assuming that $\pi_2^* > 1-z_1^2+[(1/2-z_1)^2-(1/2-z_2)^2](1+x)$, we obtain
$\underline{p}^*_2 > \frac{1-z_1^2}{1+x}+(1/2-z_1)^2-(1/2-z_2)^2$.
However, given this strategy of Firm $2$, Firm $1$ always obtains a profit $\pi_1 > 1-z_1^2=\pi_1^*$.
This is because Firm $1$ can always charge a price of $p_1 \in (\frac{1-z_1^2}{1+x},\underline{p}^*_2+(1/2-z_2)^2-(1/2-z_1)^2)$ to $C_3$ consumers,
which is strictly higher than $\frac{1-z_1^2}{1+x}$ and lower than the lower bound of Firm $2$'s price $\underline{p}^*_2+(1/2-z_2)^2$, which includes transportation costs. Thus, Firm $1$ increases its profit when it deviates to a price $p_1$ in this open interval. This is contradictory to $\pi_1^*$ being the maximum profit.

Thus, we determine that if $\pi_1^*=1-z_1^2$, then $\pi_2^* \le 1-z_1^2+[(1/2-z_1)^2-(1/2-z_2)^2](1+x)$.

Next, we show that it is not the case that $\pi_2^* < 1-z_1^2+[(1/2-z_1)^2-(1/2-z_2)^2](1+x)$.
Assuming that $\pi_2^* < 1-z_1^2+[(1/2-z_1)^2-(1/2-z_2)^2](1+x)$, if Firm $2$ chooses
\begin{align*}
\dot{p}_2 = \frac{1-z_1^2}{1+x}+(1/2-z_1)^2-(1/2-z_2)^2-\varepsilon,
\end{align*}
it obtains $C_3$ because $\dot{p}_2$ satisfies $\underline{p}_1^*+(1/2-z_1)^2 > \dot{p}_2+(1/2-z_2)^2$ and $y-(\dot{p}_2+(1/2-z_2)^2)=y-\frac{1-z_1^2}{1+x}-(1/2-z_1)^2+\varepsilon > 0$ is satisfied. Here,
by \eqref{tomonitikai}, $\varepsilon > 0$ is chosen as being satisfied by $y-\frac{1-z_1^2}{1+x}-(1/2-z_1)^2+\varepsilon > 0$.

Now, by $x \ge 1$, we obtain $\frac{1-z_1^2}{1+x}+(1/2-z_1)^2 \le 1-(1-z_2)^2+(1/2-z_2)^2$. Because the left-hand side of the equation is at most $3/4$, the minimum of the right-hand side is $3/4$.
Thus, $C_2$ purchases Firm $2$'s product because $1-(1-z_2)^2-\dot{p}_2 \ge \varepsilon > 0$.

Therefore, we obtain Firm $2$'s profit as follows,
\begin{align*}
\dot{\pi}_2=\dot{p}_2(1+x)={1-z_1^2}+[(1/2-z_1)^2-(1/2-z_2)^2](1+x)-\varepsilon(1+x).
\end{align*}
However, this contradicts the definition that $\pi_2^*$ is the maximum profit because an $\varepsilon$ exists such that
$\pi_2^* < \dot{\pi}_2$ when we take a small enough $\varepsilon$.

Thus, we obtain $\pi_2^* = 1-z_1^2+[(1/2-z_1)^2-(1/2-z_2)^2](1+x)$ if $\pi_1^*=1-z_1^2$.
\end{proof}

\section{Another proofs}
\begin{proof}[Proof of Proposition~\ref{toitoi1-z1^2+rent}]
We show that for any of given Firm $2$'s mixed strategies, Firm $1$ gains a constant expected profit when it adopts a strategy that belongs to the support of its strategy.
Given $F^{2*}(p)$ of Firm $2$, Firm $1$ charges $p_1 \in [\frac{1-z_1^2}{1+x},y-(1/2-z_1)^2)$, Firm $1$ gains an expected profit
\begin{align*}
p_1+p_1x(1-F^{2*}(p_1+T_1))=p_1+p_1x\left(\frac{\pi_1^*-p_1}{p_1x}\right)=\pi_1^*
\end{align*}
because Firm $1$ obtains $C_3$ consumers if Firm $2$ charges $p$ such that $p > p_1+T_1$ holds.
Similarly, Firm $2$'s expected profit is $\pi_2^*$ when it charges $p_2 \in [\frac{1-z_1^2}{1+x}+T_1,y-(1/2-z_2)^2]$
because Firm $2$ obtains $C_3$ consumers if Firm $1$ charges $p$ such that $p > p_2-T_1$ holds.

Now, we show that \eqref{toitoi1}-\eqref{toitoi2} is an equilibrium strategy profile.
By $\pi_1^* = 1-z_1^2$, Firm $1$ does not deviate to a price strictly lower than $\frac{1-z_1^2}{1+x}$.
$\frac{1-z_1^2}{1+x}$ is the lower bound of Firm $1$'s support of its strategy. Thus Firm $2$ never charges a price
strictly lower than $\frac{1-z_1^2}{1+x}+T_1$, which is the lower bound of Firm $2$'s strategy support.

Next, we show that Firm $1$ never deviates to a price $\dot{p}_1$ such that $y -(1/2-z_1)^2 < \dot{p}_1 < 1-z_1^2$ holds.
If Firm $1$ charges $\dot{p}_1$, then, it obtains only $C_1$ consumers and cannot obtain $C_3$ consumers because $\dot{p}_1$ has exceeded the reservation price of $C_3$ consumers. Thus, Firm $1$ obtains a strictly lower profit than $1-z_1^2$ if it charges $\dot{p}_1$
Therefore, Firm $1$ does not deviate.

Next, we show that Firm $1$ never deviates to a price $\ddot{p}_1=y-(1/2-z_1)^2$. If Firm $1$ charges $\ddot{p}_1$, it does not obtain
$C_3$ consumers except when Firm $2$ charges $p_2=y-(1/2-z_2)^2$. In this case, the $C_3$ market is divided equally between both firms.
Thus, we obtain Firm $1$'s expected profit as follows.
\begin{align*}
\ddot{p}_1\left(1+\frac{x}{2}\left(\frac{\pi_1^*-\ddot{p}_1}{\ddot{p}_1x}\right)\right)
=\frac{1}{2}\left(\ddot{p}_1+\pi_1^*\right).
\end{align*}
However, by $z_1+z_2 \le 1$ and $y-(1/2-z_2)^2 < 1-(1-z_2)^2$, we obtain $\frac{1}{2}\left(\ddot{p}_1+\pi_1^*\right)-\pi_1^*=\frac{1}{2}\left(\ddot{p}_1-\pi_1^*\right) < 0$. Thus, Firm $1$ never deviates to a price $\ddot{p}_1=y-(1/2-z_1)^2$.

In the same way, we show that Firm $2$ does not deviate to a price other than $p_2 \in [\frac{1-z_1^2}{1+x}+T_1,y-(1/2-z_2)^2]$.
\end{proof}

\begin{proof}[Proof of Proposition~\ref{mixmixalpha1}]
It is sufficient to show Firm $1$'s case because of symmetry. We show that given $G_2^{**}(z_2)$, Firm $1$ can gain an expected profit equal to $1$ on either $z_1= 0$ or all $z_1 \in [1-\hat{z}_2,1/2]$ intervals.

Firm $1$ can gain a profit equal to $1$ at $z_1=0$. Thus, it is sufficient to show that Firm $1$ gains an expected profit equal to $1$ when it is located at any $z_1 \in [1-\hat{z}_2,1/2]$ interval.
Substituting $G_2^{**}$ with \eqref{firm1EXPECTEDPAYOFF}, using partial integration, we have
\begin{align*}
E[\Pi_1^*]&=(1-z_1^2)[G_2^{**}(z_2)]^{1-z_1}_{1/2}\notag\\
&+\left[[1-(1-z_2)^2+\{(1/2-z_2)^2-(1/2-z_1)^2\}(1+x)]G_2^{**}(z_2)\right]^{\hat{z}_2}_{1-z_1}\notag\\
&-\int^{\hat{z}_2}_{1-z_1}[1-(1-z_2)^2+\{(1/2-z_2)^2-(1/2-z_1)^2\}(1+x)]^{\prime}G_2^{**}(z_2)dz_2\notag\\
&+\left\{1-G_2^{**}(\hat{z}_2)\right\}\left\{y-(1/2-z_1)^2\right\}(1+x)\\
=&[1-(1-\hat{z}_2)^2+\{(1/2-\hat{z}_2)^2-(1/2-z_1)^2\}(1+x)]G_2^{**}(\hat{z}_2)\\
&-\int^{\hat{z}_2}_{1-z_1}\{(2z_2-1)x+1\}G_2^{**}(z_2)dz_2\\
&+\left\{1-G_2^{**}(\hat{z}_2)\right\}\left\{y-(1/2-z_1)^2\right\}(1+x)\\
=&[1-(1-\hat{z}_2)^2]G_2^{**}(\hat{z}_2)\\
&+\left\{1-G_2^{**}(\hat{z}_2)\right\}\left\{y-(1/2-z_1)^2\right\}(1+x)\\
&+[(1/2-\hat{z}_2)^2-(1/2-z_1)^2](1+x)G_2^{**}(\hat{z}_2)\\
&-(1+x)[(z_2-1/2)^2]^{\hat{z}_2}_{1-z_1}\\
=&[1-(1-\hat{z}_2)^2]G_2^{**}(\hat{z}_2)\\
&+\left\{1-G_2^{**}(\hat{z}_2)\right\}\left\{y-(1/2-z_1)^2\right\}(1+x)\\
&-\left\{1-G_2^{**}(\hat{z}_2)\right\}[(1/2-\hat{z}_2)^2-(1/2-z_1)^2](1+x)\\
=&[1-(1-\hat{z}_2)^2]G_2^{**}(\hat{z}_2)+\{1-G_2^{**}(\hat{z}_2)\}\left\{y-(\hat{z}_2-1/2)^2\right\}(1+x)=1.
\end{align*}
Thus, the expected profit is constant in support of an equilibrium.

Finally, we show that for any given strategy of Firm $2$, Firm $1$ does not choose $z_1 \in (0,\hat{z}_1)$.
Given $G_2^{**}(\hat{z}_2)$, Firm $1$ does not deviate to $z_1 \in (0,\bar{z}_1]$.
This is because Firm $1$ gains $1-\dot{z}_1^2 < 1$ when it is located on $\dot{z}_1 \in (0,\bar{z}_1]$.

Given $G_2^{**}(\hat{z}_2)$, Firm $1$ does not choose $z_1 \in (\bar{z}_1,\hat{z}_1)$ when it uses a pure strategy.
Note that when $1-\hat{z}_2=\hat{z}_1$ holds, we have
\begin{align*}
(1-\hat{z}_1^2)G_2^{**}(\hat{z}_2)+\{1-G_2^{**}(\hat{z}_2)\}\left\{y-(\hat{z}_1-1/2)^2\right\}(1+x)=1.
\end{align*}
Let $f(z_1)$ denote a firm $1$'s profit when firm $1$ chooses $z_1 \in (\bar{z}_1,\hat{z}_1)$. We have
\begin{align}\label{hatz1def}
f(z_1)=(1-z_1^2)G_2^{**}(\hat{z}_2)+\{1-G_2^{**}(\hat{z}_2)\}\left\{y-(z_1-1/2)^2\right\}(1+x).
\end{align}
When derivatives of \eqref{hatz1def} are taken with respect to $z_1$, we have
\begin{align*}
f^\prime(z_1)=-2z_1G_2^{**}(\hat{z}_2)+\{1-G_2^{**}(\hat{z}_2)\}(1-2z_1)(1+x).
\end{align*}
We determine that $f^{\prime\prime}=-2-2(1-G)x < 0$, and we have\begin{align*}
f^\prime(\hat{z}_1)=-2\hat{z}_1G_2^{**}(\hat{z}_2)+\{1-G_2^{**}(\hat{z}_2)\}(1-2\hat{z}_1)(1+x)=0.
\end{align*}

Thus, we determine that both $f^\prime \ge 0$ and $f^{\prime\prime} < 0$ hold for all $z_1 \in [\bar{z}_1,\hat{z}_1]$.
\eqref{hatz1def} is monotone increasing. Therefore, Firm $1$ deviates to the left-hand side of $\hat{z}_1$.
\end{proof}

\end{document}